\documentclass[a4paper,american,english]{article}
\usepackage[T1]{fontenc}
\usepackage[latin9]{inputenc}
\usepackage[pdftex]{xcolor}
\usepackage{pdfcolmk}
\usepackage{float}
\usepackage{amsthm}
\usepackage{amssymb}
\usepackage[pdftex]{graphicx}
\PassOptionsToPackage{normalem}{ulem}
\usepackage{ulem}

\makeatletter

\pdfpageheight\paperheight
\pdfpagewidth\paperwidth

\floatstyle{ruled}
\newfloat{algorithm}{tbp}{loa}
\providecommand{\algorithmname}{Algorithm}
\floatname{algorithm}{\protect\algorithmname}
\providecolor{lyxadded}{rgb}{0,0,1}
\providecolor{lyxdeleted}{rgb}{1,0,0}

  \theoremstyle{definition}
  \newtheorem{defn}{\protect\definitionname}
  \theoremstyle{plain}
  \newtheorem{prop}{\protect\propositionname}

\usepackage{cite}

\usepackage{algorithmic}

\usepackage[cmex10]{amsmath}
\makeatother

\usepackage{babel}
  \addto\captionsamerican{\renewcommand{\definitionname}{Definition}}
  \addto\captionsamerican{\renewcommand{\propositionname}{Proposition}}
  \addto\captionsenglish{\renewcommand{\definitionname}{Definition}}
  \addto\captionsenglish{\renewcommand{\propositionname}{Proposition}}
  \providecommand{\definitionname}{Definition}
  \providecommand{\propositionname}{Proposition}
\addto\captionsamerican{\renewcommand{\algorithmname}{Algorithm}}

\begin{document}

\title{Low-complexity dominance-based Sphere Decoder for MIMO Systems%
\thanks{The material in this paper was presented in part at the Sixth International
Symposium on Wireless Communication Systems 2010 (ISWCS'10), University
of York, York, UK. September, 19--22, 2010.%
}}

\author{Gianmarco Romano, Domenico Ciuonzo, \\
Pierluigi Salvo~Rossi, Francesco Palmieri%
\thanks{The authors are with the Department of Industrial and Information
Engineering, Second University of Naples, via Roma, 29, 81031 Aversa
(CE), Italy. Email: \texttt{\{gianmarco.romano, domenico.ciuonzo,
pierluigi.salvorossi, francesco.palmieri\}@unina2.it}.%
}}

\date{{}}
\maketitle
\begin{abstract}
The sphere decoder (SD) is an attractive low-complexity alternative
to maximum likelihood (ML) detection in a variety of communication
systems. It is also employed in multiple-input multiple-output (MIMO)
systems where the computational complexity of the optimum detector
grows exponentially with the number of transmit antennas. We propose
an enhanced version of the SD based on an additional cost function
derived from conditions on worst case interference, that we call dominance
conditions. The proposed detector, the king sphere decoder (KSD),
has a computational complexity that results to be not larger than
the complexity of the sphere decoder and numerical simulations show
that the complexity reduction is usually quite significant.
\end{abstract}

\section{Introduction}

Currently, system design for wireless communications assumes the presence
of multiple antennas at both transmit and receive locations in order
to meet the requirements for high data rate transmission \cite{Biglieri2007}.
The main reason is found in the equivalent multiple-input multiple-output
(MIMO) channel providing diversity and/or capacity gains to the system,
where in the last case, compared to single-antenna systems, capacity
is increased by a factor equal to the minimum number of transmit and
receive antennas. 

The problem of (optimal) maximum-likelihood (ML) decoding in MIMO
systems is known to be exponentially complex in the number of transmit
antennas \cite{Verdu1998,Proakis2000}. Various suboptimal algorithms
have been developed as low-complexity alternatives to ML decoding,
e.g. branch and bound techniques \cite{Luo2004}, lattice-based approaches
\cite{Mow2003} and other tree-search algorithms as the A{*} algorithm
\cite{Ekroot1996}. A comprehensive study highlighting the connections
among various approaches for low-complexity ML decoding in wireless
communications is found in \cite{Murugan2006}.

In the framework of communication and information theory, the term
s\emph{phere decoder }(SD) usually refers to a collection of extremely
efficient algorithms based on number-theoretic tools, providing optimal
or nearly-optimal solutions with reduced average computational complexity
with respect to the exhaustive search of standard ML decoding. Inspired
from the work on vector search in lattices \cite{Fincke1985,Schnorr1994},
various SD algorithms have been proposed, e.g. for ML sequence estimation
in channels with memory\cite{Mow1994} and ML decoding for multidimensional
modulations in fading channels \cite{Viterbo1999}. SD has been then
extended in the context of multiantenna systems, both for uncoded
and space-time coded transmissions\cite{Damen2000}. Description and
performance comparison of different methods for SD-based ML decoding
are found in \cite{Agrell2002,Damen2003}: both works conclude that
Schnorr--Euchner-based SD (SESD) outperforms other SD variants. Furthermore,
the limitation of the algorithm to underloaded scenarios, i.e. with
number of transmit antennas not exceeding the number of receive antennas,
has been tackled in successive works dealing with optimal decoding
in (underdetermined) overloaded systems \cite{Cui2005,Chang2007,Wong2007}.
It is worth noticing that some works showed that the expected complexity
of SD is polynomial for a wide range of number of antennas and signal-to-noise
ratio (SNR) values \cite{Hassibi2005,Vikalo2005}, however according
to a more rigorous definition of expected complexity other works state
that SD exhibits reduced (w.r.t. ML) exponential complexity \cite{Jalden2005}.
Other SD algorithms approaching near-ML performance and suitable for
implementation with very large scale integration (VLSI) architectures
have been proposed in \cite{Guo2006}.

A different approach for ML decoding, based on dominance conditions,
has been studied in \cite{Odling2000,Axehill2008,Romano2009} for
systems adopting BPSK or QPSK modulation, and then extended in \cite{Romano2010}
to arbitrary-size PSK modulation. Such an algorithm, namely \emph{king
decoder }(KD), provides the ML solution and thus it is optimal from
the point of view of Symbol Error Rate (SER) performance. Two major
advantages are: (i) no matrix inversion and/or factorization is needed;
(ii) the same algorithm applies to both underloaded and overloaded
systems.

The main contribution of this paper\textcolor{red}{{} }is an enhanced
version of SD, which is based on an additional cost function derived
from dominance conditions, thus exploiting the properties of KD. The
new algorithm presents a significantly reduced computational complexity,
measured as the average number of visited nodes, w.r.t. the classic
SD. 

The rest of the paper is organized as follows: in Section~\ref{sec:system-model}
we present the mathematical model for the system under investigation;
Section~\ref{sec:sphere-decoder} describes the SD; dominance conditions,
representing the core of the improving innovation, are analytically
studied in Section~\ref{sec:dom-cond}; the proposed KSD for MIMO
detection is described in Section~\ref{sec:KSD}; in Section~\ref{sec:sim-results}
we show and compare the performance in terms of computational complexity
obtained via numerical simulations; finally, concluding remarks are
given in Section~\ref{sec:conclusions}.

\emph{Notation} - Lower-case bold letters denote vectors, with $a_{n}$
denoting the $n$th entry of $\mathbf{a}$; upper-case bold letters
denote matrices, with $a_{n,m}$ and $\mathbf{a}_{m}$ denoting the
$\left(n,m\right)$th entry and the $m$th column of $\mathbf{A}$,
respectively; $\mathbb{E}\left\{ \cdot\right\} $, $\left(\cdot\right)^{*}$,
$\left(\cdot\right)^{T}$, $\left(\cdot\right)^{H}$, and $\left\Vert \cdot\right\Vert _{2}$,
denote expectation, conjugate, transpose, conjugate-transpose and
squared Frobenius norm operators, respectively.

\section{System Model}

\label{sec:system-model}

We consider a narrowband MIMO system with $K$ transmit antennas and
$N$ receive antennas, described by the following vector model 
\begin{equation}
\mathbf{y}=\mathbf{H}\mathbf{x}+\mathbf{n},\label{eq:mimo-system-model}
\end{equation}
where $\mathbf{y}\in\mathbb{C}^{N}$ is the received vector, whose
entry $y_{i}$ represents the signal received by the $i$th receive
antenna; $\mathbf{H}\in\mathbb{C}^{N\times K}$ is the channel matrix,
whose entry $h_{ij}$ represents the fading coefficient between the
$j$th transmit antenna and the $i$th receive antenna; $\mathbf{x}\in\mathbb{C}^{K}$
is the transmitted vector, whose entry $x_{j}$ represents the symbol
transmitted by the $j$th transmit antenna; $\mathbf{n}\in\mathbb{C}^{N}$
is the additive noise vector modeled according to a zero-mean complex
Gaussian distribution with variance $\mathbb{E}\left\{ \mathbf{n}\mathbf{n}^{H}\right\} =\eta_{0}\mathbf{I}_{N}$.
Transmitted symbols are drawn from a finite set of complex symbols
$\chi$ which depends on the specific chosen modulation scheme. The
channel vector from the $k$th transmit antenna is $\mathbf{h}_{k}$,
i.e. the $k$th column of channel matrix. Also, we assume perfect
channel state information at the receiver. 

The problem of optimal decoding $\mathbf{x}$ from the knowledge of
$\mathbf{y}$ is formulated as follows 
\begin{equation}
\mathbf{x}_{ML}=\arg\min_{\mathbf{x}\in\chi^{K}}\left\Vert \mathbf{y}-\mathbf{H}\mathbf{x}\right\Vert _{2}\label{eq:ml}
\end{equation}
where exhaustive search is apparently prohibitive for sizes of interest,
thus the need for low-complexity alternatives. Assuming the constraint
that the total average energy to be transmitted over the single symbol
period cannot exceed $E_{x}$, system performance are evaluated with
respect to the SNR per single receive antenna, i.e. ${\rm SNR}\triangleq E_{x}/\eta_{0}$.

It is worth noticing that other kinds of systems for multiuser communications,
such as direct-sequence code-division-multiple-access (DS-CDMA) \cite{Verdu1998}
and multi-carrier code-division-multiple-access (MC-CDMA) \cite{Hanzo2006},
share the same linear model with additive noise described by \eqref{eq:mimo-system-model}.

\section{Sphere Decoder}

\label{sec:sphere-decoder}

The idea of sphere decoding is to restrict the search to transmitted
vectors whose received constellation counterparts are included in
a hyper-sphere with radius $r$ centered on the received signal $\mathbf{y}$,
that is 
\begin{equation}
\left\Vert \mathbf{y}-\mathbf{H}\mathbf{x}\right\Vert ^{2}<r^{2}.\label{eq:sphere_decoding_condition}
\end{equation}
If the sphere contains no vectors the algorithm either fails or restarts
with an increased radius. In the latter case the result of the algorithm
is always the optimal ML solution, obtained with reduced computational
complexity when the number of vectors in the sphere is small compared
to the overall number of possible transmitted vectors, i.e. $\mathcal{\left|\chi\right|}^{K}$.
The choice of the radius is crucial in order to obtain a computational
complexity gain; in the ideal case the sphere should include just
one vector.

The test in \eqref{eq:sphere_decoding_condition} is efficiently performed
by exploiting the $\mathbf{QL}$ (corresp. $\mathbf{QR}$) factorization
of the channel matrix $\mathbf{H}$ in terms of a unitary matrix $\mathbf{Q}$
(i.e. $\mathbf{Q}^{H}\mathbf{Q=\mathbf{I}}_{N})$ and a lower-triangular
matrix $\mathbf{L}$ (corresp. upper-triangular matrix $\mathbf{R}$).
In this case \eqref{eq:ml} can be equivalently formulated as 
\begin{align}
\mathbf{x}_{ML} & =\arg\min_{\mathbf{x}\in\chi^{K}}\left\Vert \tilde{\mathbf{y}}-\mathbf{L}\mathbf{x}\right\Vert ^{2}\label{eq:equiv-ml}\\
 & =\arg\min_{\mathbf{x}\in\chi^{K}}\sum_{i=1}^{K}\left|\tilde{y}_{i}-\sum_{j=1}^{i}l_{ij}x_{j}\right|^{2},\label{eq:ed}
\end{align}
where $\tilde{\mathbf{y}}\triangleq\mathbf{Q}^{T}\mathbf{y}$. The
QR factorization enables the test in \eqref{eq:sphere_decoding_condition}
to be formulated as a tree search with pruning\cite{Murugan2006}.
In fact the summation in \eqref{eq:ed} can be performed on a tree
with $K+1$ layers where the term 
\begin{equation}
\left|\tilde{y}_{i}-\sum_{j=1}^{i}l_{ij}x_{j}\right|^{2},\label{eq:pd}
\end{equation}
can be computed at each node of the layer $i$. The advantage of this
formulation is that the partial distance in \eqref{eq:pd} is always
positive; this fact implies that the children nodes have always greater
partial distances, i.e. the metric is said to be \emph{cumulative}.
Therefore at each node at layer $i$ we can compute the accumulated
partial distance 
\begin{equation}
\sum_{k=1}^{i}\left|\tilde{y}_{i}-\sum_{j=1}^{i}l_{ij}x_{j}\right|^{2},\label{eq:apd}
\end{equation}
and compare it with a threshold, corresponding to $r^{2}$. The algorithm
selects only the nodes leading to leaves that are within a sphere
and at the same time computes the metric that will be used at the
end to select the optimal solution. As stated before, if no leaves
are contained in the sphere then the radius is increased and the search
on the tree is restarted.

There are two possible strategies to perform the tree search: the
breadth-first search (BFS) and the depth-first search (DFS) \foreignlanguage{american}{\cite{Murugan2006}}.
In the breadth-first search, all surviving nodes of the same level
are visited before moving to the next level, until the leaves are
reached. In the depth-first, at each level only one node is visited,
and following its child in $K$ steps a leaf is reached. At this point
the radius is updated and the algorithm proceeds with other nodes
starting from upper levels. While in BFS the tree is traversed from
top to bottom, in DFS the tree is traversed horizontally. In the latter
case the algorithm can be started with an infinite radius as it can
be updated as soon as the first leave is reached. The performance
of the SD algorithm can be improved by choosing a proper enumeration
order. In Fincke-Pohst enumeration \cite{Fincke1985} branches are
enumerated in a natural fashion, while in Schnorr-Euchner enumeration
\cite{Schnorr1994,Agrell2002} branches are selected in a \emph{zig-zag}
fashion for QAM constellations along each dimension \cite{Damen2003}.

The computational complexity of the sphere decoding algorithm is measured
by the average number of visited nodes needed to obtain (\ref{eq:equiv-ml})
\cite{Jalden2005}. That figure is closely related to the time required
by the algorithm to provide the solution and clearly related to the
throughput that is achievable in currently available digital hardware
\cite{Studer2008}. The computational cost of $\mathbf{QL}$ factorization
is not considered here, since it is computed once for all and it represents
a negligible factor in the overall complexity.

\section{Dominance Conditions}

\label{sec:dom-cond}

In the sphere decoding algorithm at each node the partial distance
is checked in order to exclude some branches in the tree. Another
condition can be derived from the Euclidean distance that can improve
the computational complexity of the sphere decoder. In this section
we derive a set of sufficient conditions that can be used to exclude
some possible transmitted vectors from the set of candidates in the
ML search. 

Geometrically the ML solution is given by the vector $\mathbf{x}$
that minimizes the Euclidean distance 
\begin{equation}
f\left(\mathbf{x}\right)=\left(\mathbf{y}-\mathbf{H}\mathbf{x}\right)^{H}\left(\mathbf{y}-\mathbf{H}\mathbf{x}\right).\label{eq:euclidean-distance}
\end{equation}
We first define the difference of the Euclidean distance between two
generic points of $\chi^{K}$. 
\begin{defn}
\label{def:discrete-difference}Given two generic vectors $\mathbf{x}$
and $\hat{\mathbf{x}}$, with $\{\mathbf{x},\hat{\mathbf{x}}\}\in\chi^{K}$,
the \emph{discrete difference} is defined as $\Delta f\left(\mathbf{x};\hat{\mathbf{x}}\right)\triangleq f\left(\mathbf{x}\right)-f\left(\hat{\mathbf{x}}\right)$
. 
\end{defn}
\begin{defn}
The discrete difference related to vectors differing only in the $k$th
component is called \emph{$k$th discrete difference} along the $k$th
coordinate and denoted $\Delta_{k}f\left(\mathbf{x};\hat{\mathbf{x}}\right)$. 
\end{defn}
A necessary and sufficient condition for $\mathbf{x}$ to be a global
minimum for the cost function $f\left(\mathbf{x}\right)$ is then
that all discrete differences $\Delta f\left(\mathbf{x};\hat{\mathbf{x}}\right)$
are non positive for each $\hat{\mathbf{x}}\in\chi^{K}$. The search
of the global minimum just by looking at the differences does not
reduce the computational complexity of the ML search alone. The number
of differences to compute is still exponential with the number of
inputs and the size of the constellation. However, as it will be clearer
in the following, we can avoid to look at all differences and still
get the optimal solution.

In the special case of the Euclidean distance the discrete difference
along the generic $k$th coordinate takes on a specific expression,
as stated by the following proposition. 
\begin{prop}
\label{pro:k-ddiff}For any pair of vectors $\mathbf{x}$ and $\hat{\mathbf{x}}$
that belong to $\chi^{K}$ and differ only in the $k$th position
\begin{multline}
\Delta_{k}f\left(\mathbf{x};\hat{\mathbf{x}}\right)=-2\Re\left\{ \left(x_{k}-\hat{x}_{k}\right)^{*}\left[\mathbf{h}_{k}^{H}\mathbf{y}-\sum_{i\neq k}x_{i}\mathbf{h}_{k}^{H}\mathbf{h}_{i}\right]\right\} \\
+\left(\left|x_{k}\right|^{2}-\left|\hat{x}_{k}\right|^{2}\right)\mathbf{h}_{k}^{H}\mathbf{h}_{k}.\label{eq:ddiff}
\end{multline}
 \end{prop}
\begin{proof}
See appendix \ref{proof:k-ddiff}.
\end{proof}
The $k$th discrete difference in \eqref{eq:ddiff} depends on the
observed vector $\mathbf{y}$ and on the symbols of the other elements
of the input vector $\mathbf{x}$, i.e. $x_{i}$, $i\neq k$. The
sign of the discrete difference $\Delta_{k}f\left(\mathbf{x};\hat{\mathbf{x}}\right)$
determines which of the two possible transmit vectors $\mathbf{x}$
and $\hat{\mathbf{x}}$ is closer to the observation $\mathbf{y}$. 

The discrete difference expression in \eqref{eq:ddiff} can be simplified
if a constellation with constant modulus is employed, as the second
term on the right hand side of \eqref{eq:ddiff} becomes zero. In
this case the discrete difference reduces to 
\begin{multline}
\Delta_{k}f\left(\mathbf{x};\hat{\mathbf{x}}\right)=\\
-2\Re\left\{ \left(x_{k}-\hat{x}_{k}\right)^{*}\left[\mathbf{h}_{k}^{H}\mathbf{y}-\sum_{i\neq k}x_{i}\mathbf{h}_{k}^{H}\mathbf{h}_{i}\right]\right\} .\label{eq:k-ddiff-psk}
\end{multline}

\subsection{Dominance conditions for $4$-QAM}

Since $4$-QAM constellations are separable, we can equivalently consider
a real-valued system model, whose dimensions are doubled, with binary
signaling, i.e. $\chi=\left\{ -1,+1\right\} $. In the following,
theoretical results will be derived referring to the real-valued system
model. In this case, the $k$th discrete difference is
\begin{equation}
\Delta_{k}f\left(\mathbf{x};\hat{\mathbf{x}}\right)=-2\left(x_{k}-\hat{x}_{k}\right)\left[\mathbf{h}_{k}^{T}\mathbf{y}-\sum_{i\neq k}x_{i}\mathbf{h}_{k}^{T}\mathbf{h}_{i}\right].\label{eq:k-ddiff-bpsk}
\end{equation}
Eq.~ \eqref{eq:k-ddiff-bpsk} can be used to make an optimal decision
under the assumption that the contribution due to the other components
of vector $\mathbf{x}$ are known. From \eqref{eq:k-ddiff-bpsk},
a necessary condition can be derived for BPSK constellations as follows.
The discrete difference is non positive when the two terms on the
right-hand side of \eqref{eq:k-ddiff-bpsk} have the same sign
\begin{equation}
\mathrm{sign}\left(x_{k}-\hat{x}_{k}\right)=\mathrm{sign}\left[\mathbf{h}_{k}^{T}\mathbf{y}-\sum_{i\neq k}x_{i}\mathbf{h}_{k}^{T}\mathbf{h}_{i}\right].
\end{equation}
Note however that in the binary case there exists only one adjacent
point, i.e. $\hat{x}_{k}=-x_{k}$, and the above equation can be written
as
\begin{equation}
\mathrm{sign}\left(2x_{k}\right)=\mathrm{sign}\left[\mathbf{h}_{k}^{T}\mathbf{y}-\sum_{i\neq k}x_{i}\mathbf{h}_{k}^{T}\mathbf{h}_{i}\right],
\end{equation}
that can be equivalently rewritten as
\begin{equation}
x_{k}=\mathrm{sign}\left[\mathbf{h}_{k}^{T}\mathbf{y}-\sum_{i\neq k}x_{i}\mathbf{h}_{k}^{T}\mathbf{h}_{i}\right].\label{eq:equil-condition}
\end{equation}
From \eqref{eq:equil-condition} we have that the ML solution must
satisfy the set of equations
\begin{equation}
x_{k}^{ML}=\mathrm{sign}\left[\mathbf{h}_{k}^{T}\mathbf{y}-\sum_{i\neq k}x_{i}^{ML}\mathbf{h}_{k}^{T}\mathbf{h}_{i}\right],\quad k=1,\dots,K\label{eq:necessary-condition}
\end{equation}
which provide the set of local minima of the Euclidean distance, thus
representing a necessary condition for the ML solution, as for these
points all the $k$th discrete differences are non positive. 

It is interesting to note that the same set of equations has been
derived in the context of Hopfield neural network (HNN) \cite{Hopfield1982,Hopfield1984}
and applied to ML decoding. In \cite{Miyajima1993,Kechriotis1993},
detectors for code division multiple access (CDMA) have been proposed
for the first time and then the idea has been further developed in
\cite{Kechriotis1996,Kechriotis1996a,Sgraja2001,Engelhart2002}. The
Eq.~\eqref{eq:equil-condition} represents the discrete-time approximation
of the equation of motion of neurons, as the metric of ML optimum
detector can be mapped to the energy function of the HNN and the ML
solution is the result of the dynamic update of \eqref{eq:equil-condition}
(see for example \cite{Engelhart2002} and references therein). Therefore
the search is based on a gradient descent algorithm that may not provide
the exact ML solution, but rather only a local minimum. Furthermore
when the updates of the discrete-time equations are done in parallel,
the solution may also present limit cycles and no convergence to a
fixed point \cite{Marcus1989}. In order to prevent the updating rule
to enter a limit cycle and to force the dynamic update through increasing
likelihood towards the global minimum, in \cite{Sun2009} a modified
HNN approach to ML decoding is proposed, leading to a family of likelihood
ascent sub-optimal detectors (LAS). All these algorithms are sub-optimal
and can approach optimal performances only under specific conditions.

The necessary conditions in \eqref{eq:necessary-condition} suggest
to restrict the search for the ML solution to the set of local minima.
Unfortunately, no method is known to enumerate all equilibrium points,
i.e. points that satisfy \eqref{eq:necessary-condition} with a computational
complexity that it is not exponential. However, we can still identify
cases where the determination of the sign of the $k$th discrete difference,
i.e. the determination of the $k$th component of local minima, can
be made regardless of the contribution of all other components of
$\mathbf{x}$. A sufficient condition for the determination of the
sign of the $k$th discrete difference is given by the following proposition.
\begin{prop}
\label{pro:suff-cond}If the following condition is satisfied 
\begin{equation}
\left|\mathbf{h}_{k}^{T}\mathbf{y}\right|>\sum_{i\neq k}\left|\mathbf{h}_{k}^{T}\mathbf{h}_{i}\right|\label{eq:suff-cond}
\end{equation}
then the sign of the corresponding $k$th discrete difference for
BPSK constellation is determined regardless of the contribution of
all other components of $\mathbf{x}$. \end{prop}
\begin{proof}
See appendix \ref{proof:suff-cond}.
\end{proof}
Eq.~\eqref{eq:suff-cond} is a \emph{dominance condition} because,
when it holds, the $k$th component of the projected received vector
is so strong that dominates all other components. The dominance condition
assumes that in \eqref{eq:suff-cond} no symbols $x_{i}$, $i\neq k$,
are known. However, in sequential decoding, partial knowledge may
be available. In such cases the sign of the discrete difference depends
only on the subset of $x_{i}$ that are still to be decoded. A dominance
condition when only a subset $\mathcal{W}$ of symbols is already
available, can be given.
\begin{prop}
\label{pro:cond-suff-cond}Given the set of known symbols $\mathcal{W}$
and a set of unknown symbols $\mathcal{O}$, if the following condition
holds 
\begin{equation}
\left|\mathbf{h}_{k}^{T}\mathbf{y}-\sum_{m\in\mathcal{W},m\neq k}x_{m}\mathbf{h}_{k}^{T}\mathbf{h}_{m}\right|>\sum_{i\in\mathcal{O},i\neq k}\left|\mathbf{h}_{i}^{T}\mathbf{h}_{k}\right|,\label{eq:cond-suff-cond}
\end{equation}
then the sign of the corresponding $k$th discrete difference for
BPSK constellation is determined regardless of the contribution of
all components of $\mathbf{x}$, $x_{i}$, $i\in\mathcal{O}$. \end{prop}
\begin{proof}
Analogous to the proof of Prop. \ref{pro:suff-cond}.
\end{proof}
Eq.~\eqref{eq:cond-suff-cond} generalizes \eqref{eq:suff-cond}:
if some antenna $i$ is not dominant over his multiantenna interference,
it may happen that it is \emph{conditionally} dominant, as the interference
by the already known bits is canceled out.

The sufficient condition in \eqref{eq:suff-cond} was also derived
in \cite{Kechriotis1996a}, where it has been used in a multiuser
detection algorithm based on Hopfield Neural Networks. Eqs.~\eqref{eq:suff-cond}
and \eqref{eq:cond-suff-cond} have also been used in \cite{Odling2000}
for maximum-likelihood sequence detection and then\emph{ }in \cite{Axehill2008}
with a preprocessing algorithm for multiuser detection. In \cite{Romano2009}
they have been used for a stand-alone tree-search algorithm for low-complexity
ML detection in spatial multiplexing MIMO systems, the king decoder.

Eqs.~\eqref{eq:suff-cond} and \eqref{eq:cond-suff-cond} are satisfied
if the off-diagonal terms of the channel correlation matrix are small
compared to the terms $\left|\mathbf{h}_{k}^{T}\mathbf{y}\right|$,
$k=1,\dots,K$. Whether the conditions are satisfied or not depends
on the received vector $\mathbf{y}$ and on the structure of the channel
or of the correlation channel matrix.

\subsection{Dominance conditions for $M$-QAM}

In the case of $M$-QAM constellations, dominance conditions can be
expressed in terms of those for $4$-QAM, when $M=2^{n}$ and $n$
is an even number, e.g. $16$-QAM. In fact such QAM constellations
can be written as weighted linear combination of $n/2$ $4$-QAMs
\cite{Cui2005}. For example, the $16$-QAM transmit vector can be
expressed as 
\begin{equation}
\mathbf{x}=\mathbf{x}_{1}+2\mathbf{x}_{2}
\end{equation}
where $\mathbf{x}_{1},\mathbf{x}_{2}$ are $4$-QAM vectors. Consequently,
the system model \eqref{eq:mimo-system-model} can be written as
\begin{equation}
\mathbf{y}=\left[\begin{array}{cc}
\mathbf{H} & 2\mathbf{H}\end{array}\right]\left(\begin{array}{c}
\mathbf{x}_{1}\\
\mathbf{x}_{2}
\end{array}\right)+\mathbf{n}\label{eq:16QAM-MIMO-model}
\end{equation}
which represents the equivalent model for $16$-QAM MIMO systems with
$K$ transmit and $N$ receive antennas in terms of $4$-QAM MIMO
system with $2K$ transmit and $N$ receive antennas. Based on this
equivalence, we can restrict our attention to $4$-QAM MIMO systems
without loss of generality.

\subsection{Dominance conditions for $M$-PSK}

It is possible to derive analogous dominance conditions in the general
case of $M$-PSK, however in this case the real-valued system model
does not hold. Dominance conditions based on the complex-valued system
model have been derived and analyzed in \cite{Romano2010}. Results
are not reported here, as they are not necessary for popular systems
supporting QAM.

\section{King Sphere Decoder}

\label{sec:KSD}

The main contribution of this paper is the integration of the conditions
\eqref{eq:suff-cond} and \eqref{eq:cond-suff-cond} in \emph{any}
sphere decoding algorithm. The idea that we propose is to use the
conditional dominance condition given by \eqref{eq:cond-suff-cond}
at each node of the decoding tree in addition to the partial distance
condition of the standard sphere decoding algorithm. The dominance
conditions, when satisfied, allow to cut branches off that cannot
correspond to the optimal solution and then reduce the number of the
visited nodes, i.e. the computational complexity of the search. The
operation of the proposed algorithm is shown with the help of Fig.~\ref{fig:tree-decoding}
that shows a decoding tree for a system with $N=5$ and $M=5$ antennas.
At each node we can check whether \eqref{eq:cond-suff-cond} is satisfied
or not. For example the node pointed by the arrow corresponds to the
dominance condition
\begin{multline}
\left|\mathbf{h}_{3}^{T}\mathbf{y}-\mathbf{h}_{3}^{T}\mathbf{h}_{2}x_{2}-\mathbf{h}_{3}^{T}\mathbf{h}_{1}x_{1}\right|_{x_{1}=1,x_{2}=1}>\\
\left|\mathbf{h}_{3}^{T}\mathbf{h}_{4}\right|+\left|\mathbf{h}_{3}^{T}\mathbf{h}_{5}\right|.
\end{multline}
If the condition is satisfied then a decision on the corresponding
bit can be made and only one of the two branches that departs from
that node is selected, and half of child nodes can be cut off. In
our example such condition is satisfied and a decision on bit 3 can
be made: $x_{3}=-1$, if $x_{1}=1$ and $x_{2}=1$ or, equivalently,
we can exclude all the vectors that have $x_{1}=1$, $x_{2}=1$, $x_{3}=1$.
At the end we obtain a set of possible ML solutions, as shown in Fig.
\ref{fig:tree-decoding}, where only 6 out of 32 paths survive.

The tree-search algorithm that makes use of the (conditional) dominance
conditions alone has already been presented in \cite{Romano2009,Romano2010},
where it has been called king decoder (KD). In general at the end
of the tree-search the selection of the optimal solution is made among
the survivors by computing the corresponding metric and then the last
step of the search involves the computation of Euclidean distances
for all survivors. In KD rather than compute the Euclidean distance
at the end of the enumeration process, a different equivalent metric,
that is cumulative and re-uses the computations done for dominance
conditions, has been introduced \cite{Romano2009}. 

\begin{figure}
\begin{centering}
\includegraphics[width=0.9\columnwidth]{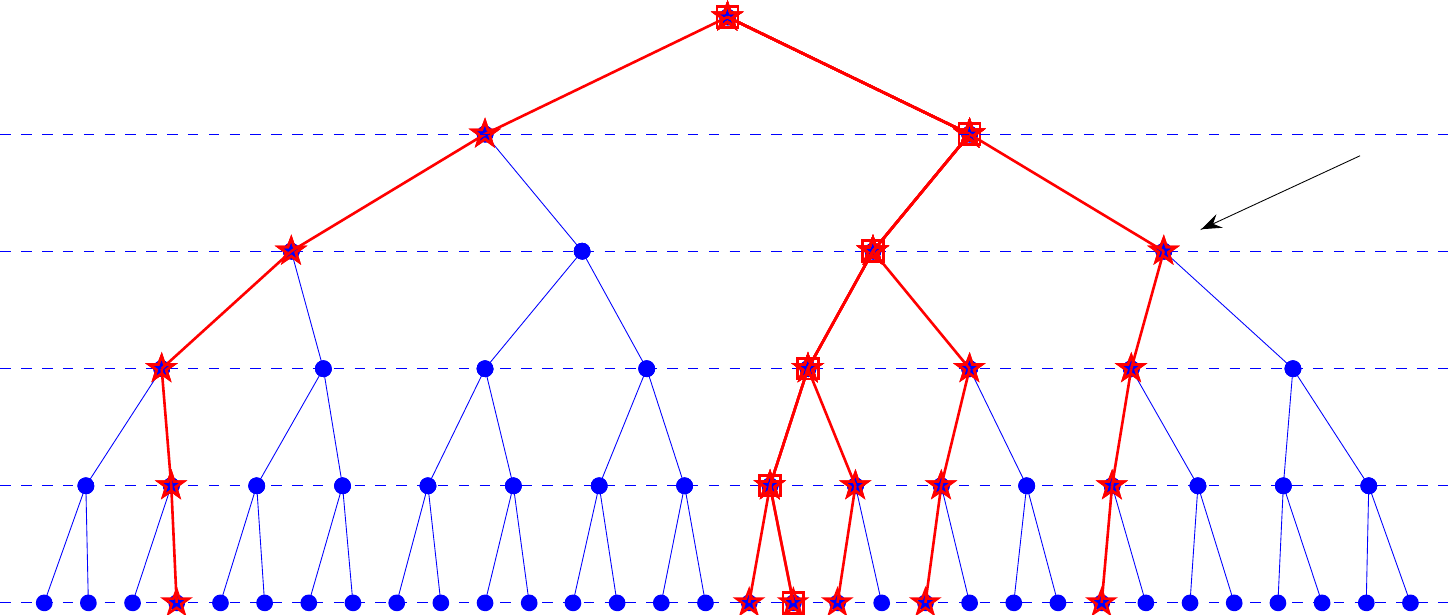}
\par\end{centering}

\caption{\label{fig:tree-decoding}Tree-search algorithm for a system with
$N=5$ and $K=5$ antennas. The transmitted bit vector is $\mathbf{x}=\left(1,-1,-1,-1,1\right)^{T}$.
Paths with stars are provided by the dominance conditions alone, while
the path with square nodes is the ML solution }
\end{figure}

We propose in this paper the inclusion of the dominance conditions
as an additional step in a generic tree-search algorithm for ML decoding.
For simplicity we restrict our attention to sphere decoding and we
show that at the expenses of a marginal increase of computational
complexity at each node, a significant reduction of the average number
of visited node can be achieved. By integrating the dominance conditions
into sphere decoding, we can exclude points that cannot be ML solution
before checking if they lie within the sphere. At the end of the tree-search
the partial metric computation carried on by the sphere decoder can
be used to select the optimal solution. We call this enhanced version
of the sphere decoder, \emph{king sphere decoder }(KSD).

We consider the formulation of a generic tree-search algorithm based
on the pseudo-code provided by Murugan \emph{et al.} in \cite{Murugan2006}
which describes a generic branch-and-bound algorithm. More generally
in a tree-search algorithm at each node a decision is made based on
a boolean condition that it is not necessarily expressed as a cost
function compared to a bounding function, but as combination of several
boolean conditions. 

\begin{algorithm}
\caption{\label{alg:gbb}Generic tree-search algorithm (adapted from \cite{Murugan2006})}
\texttt{%
\begin{algorithmic}
\STATE{reset\_tree()} \COMMENT{initialize the tree}
\STATE{init\_search()} \COMMENT{ex.: reset partial distance}
\STATE{init\_ACTIVE()} \COMMENT{Create an empty list of active nodes}
\STATE{cn = root} \COMMENT{current node (cn) is root}
\STATE{ //-- Main loop }
\WHILE{ cn is not empty }
\IF{ cn is not a leaf }
\IF{ cn is a valid node }
\STATE{ get valid child nodes of cn }
\STATE{ sort valid child nodes }
\STATE{ insert valid child nodes in ACTIVE }
\STATE{ update node counter }
\ENDIF
\ELSE
\STATE{ select best node }
\STATE{ update bounding function }
\ENDIF
\STATE{ get next node in ACTIVE }
\ENDWHILE
\STATE{ // }
\IF{ best node is empty }
\STATE{ restart with a reduced radius }
\ELSE
\STATE{ get the ML solution corresponding to the best node }
\ENDIF
\end{algorithmic}
}
\end{algorithm}

In the algorithm, \texttt{ACTIVE} contains an ordered set of nodes
to be visited. The data structure used to implement \texttt{ACTIVE}
determines the traverse strategy in the tree. In case of BFS a queue
data structure can be employed, while a stack is suitable for DFS.
The algorithm starts with the initialization of the radius, that can
also be infinite, as in DFS. The main loop visits each valid node
of the tree, starting from the root. At each node, unless a leaf is
reached, the (conditional) dominance is checked first. If it is satisfied
then one of the two child nodes can be excluded and will not be visited,
otherwise no action is taken. Then, for each child nodes that has
not been excluded, the partial distance is computed and compared against
to the current radius. At this point only nodes that lie within the
partial distance will be considered valid nodes. Therefore for each
node, in general, a sub-set of child nodes are valid node and will
be visited in the loop. Note that dominance conditions are applied
to the current node and partial distances are computed on its child
nodes only if they are not excluded by the previous check. If valid
nodes that are generated from the current node need to be sorted,
as for example in Schnorr-Euchner enumeration, a sort function is
called before nodes are inserted in \texttt{ACTIVE}. 

If the current visited node is a leaf then according to the metric,
that is cumulatively computed, then the best candidate can be chosen
and, depending on the tree traversing strategy, the radius may be
updated. If a BFS is employed then it may happen that no leaf nodes
are available at the end of the main loop (there is no best node)
and a new search must be performed with an increased radius.

Note that the only required modification with respect to the sphere
decoding algorithm is contained in the function that generates valid
child nodes.

Dominance conditions introduced in the king sphere decoder can be
seen as new set of constraints reducing the number of points to be
visited, and for which no partial distance needs to be computed, because
the new algorithm discards paths that surely cannot be local minima
and then cannot be the ML solution.

As for sphere decoder, the advantage of the KSD is the expected large
reduction of the number of the visited nodes and then of surviving
paths. In the best-case scenario, in every visited node the dominance
condition is satisfied, and then the algorithm returns a unique solution
that corresponds to the ML solution, and only $M$ nodes are visited,
regardless of the choice of radius. In general the number of visited
nodes is greater than $M$ because the condition in \eqref{eq:cond-suff-cond}
is not always satisfied. In the worst-case scenario no dominant bit
is found, and then no decrease in the number of visited nodes with
respect to the original tree-search algorithm is achieved. While the
added conditions might increase the computational complexity at each
node, the average number of visited nodes can only be decreased. Therefore
the algorithm can only perform better in terms of computational complexity
measured in terms of the average number of visited nodes at the expenses
of increased computation at each node. In practical implementations
this represents a good trade-off between speed, and then achievable
throughput, and area on VLSI devices.

The efficiency of the algorithm will depend on the structure of the
channel, i.e. on the matrix $\mathbf{H}$ and is higher in those cases
where the off-diagonals elements of the channel correlation matrix
are relatively small. This might be the case of some correlative MIMO
channel models that take into account correlation among transmit and
receive antennas or keyhole channels \cite{Chizhik2002}.

Note that the dominance conditions do not require any matrix inversion
or matrix factorization and can be employed unmodified both in underloaded
and overloaded systems.

\section{Simulation Results\label{sec:sim-results}}

The proposed algorithm has always optimal performances in terms of
SER, by construction. Performances are then measured in terms of the
average number of visited nodes. We have run Monte-Carlo simulations
in order to verify the improvement that can be gained with our proposed
algorithm as in the worst case scenario performances are the same
as those of SD. 

Simulation results are presented with reference to two typologies
of wireless channels, with different mathematical structures in their
channel matrices: (i) independent fading, where entries of the channel
matrix are assumed to be i.i.d. according to a zero-mean complex Gaussian
distribution with unit variance; (ii) correlated fading, where a Kronecker
model is assumed to take spatial correlation into account \cite{Gesbert2002}.
More specifically, in the case of correlated fading we assume that
the channel matrix follows the structure \cite{Zelst2002}
\begin{equation}
\mathbf{H}={\bf R}_{R}^{1/2}\mathbf{G}{\bf R}_{T}^{1/2},\label{eq:correlated fading}
\end{equation}
where $\mathbf{R}_{T}$ and $\mathbf{R}_{R}$ describe spatial correlation
at transmit and receive locations, respectively, and $\mathbf{G}$
matches the independent fading structure.

In Fig.~\ref{fig:avg-num-nodes} results from simulations are shown
for MIMO systems with different number of transmit and receive antennas.
Two MIMO channel models are considered. The first is the standard
MIMO channel model where the channel matrix elements are drawn from
a complex Gaussian distribution. The second model is the correlative
MIMO channel where correlation between transmit antennas and between
receive antennas as in \eqref{eq:correlated fading} where, according
the model proposed in \cite{Zelst2002}, we have
\begin{equation}
\mathbf{R}_{T}=\left(\begin{array}{ccccc}
1 & \rho_{T} & \rho_{T}^{4} & \cdots & \rho_{T}^{\left(K-1\right)^{2}}\\
\rho_{T} & 1 & \ddots & \ddots & \vdots\\
\rho_{T}^{4} & \rho_{T} & 1 & \ddots & \rho_{T}^{4}\\
\vdots & \ddots & \ddots & \ddots & \rho_{T}\\
\rho_{T}^{\left(K-1\right)^{2}} & \cdots & \rho_{T}^{4} & \rho_{T} & 1
\end{array}\right)
\end{equation}
and 
\begin{equation}
\mathbf{R}_{R}=\left(\begin{array}{ccccc}
1 & \rho_{R} & \rho_{R}^{4} & \cdots & \rho_{R}^{\left(N-1\right)^{2}}\\
\rho_{R} & 1 & \ddots & \ddots & \vdots\\
\rho_{R}^{4} & \rho_{R} & 1 & \ddots & \rho_{R}^{4}\\
\vdots & \ddots & \ddots & \ddots & \rho_{T}\\
\rho_{R}^{\left(N-1\right)^{2}} & \cdots & \rho_{R}^{4} & \rho_{R} & 1
\end{array}\right)
\end{equation}
with $\rho_{T}$ and $\rho_{R}$ transmit and receive correlation
indexes, respectively. Results are obtained for $\rho_{T}=0.5$ and
$\rho_{R}=0.5$ and both SD and KSD, with deep first (DF) search strategy,
in terms of the number of visited nodes averaged over the channel
and noise realizations as well as the possible transmitted vectors
\cite{Jalden2005}.

Figs.~\ref{fig:avg-num-nodes}, \ref{fig:avg-num-nodes-1} and \ref{fig:avg-num-nodes-2}
show that in practice dominance conditions can effectively reduce
the computational complexity of SD in all cases under consideration.
The reduction is greater with correlated MIMO systems, suggesting
that dominance conditions are more frequently satisfied in this case.

\noindent 
\begin{figure}
\noindent \begin{centering}
\includegraphics[width=1\columnwidth]{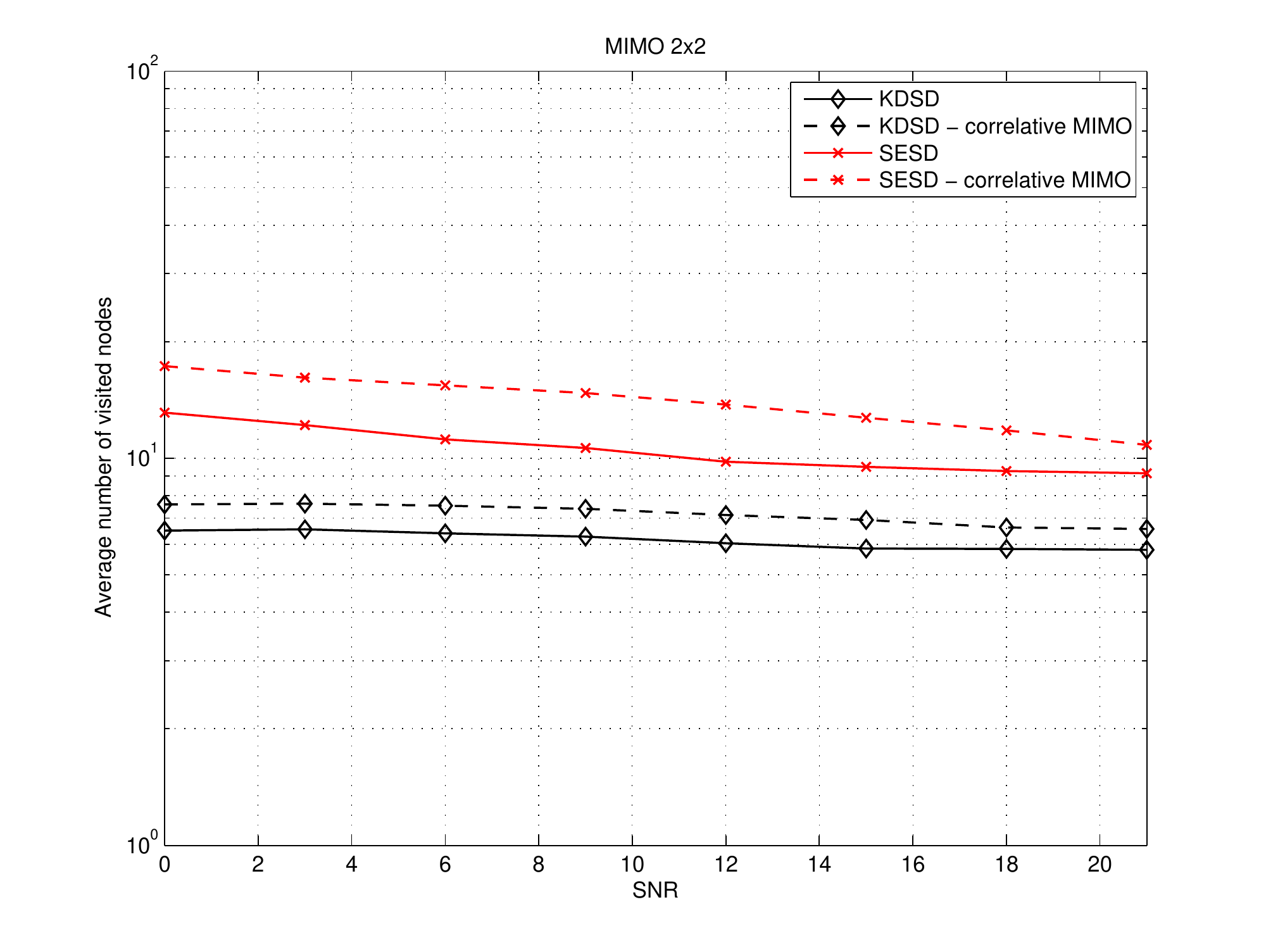}
\par\end{centering}

\caption{\label{fig:avg-num-nodes}Average number of visited nodes as function
of average signal-to-noise ratio. $4$-QAM system with $K=2$, $N=2$.}
\end{figure}

\noindent 
\begin{figure}
\noindent \begin{centering}
\includegraphics[width=1\columnwidth]{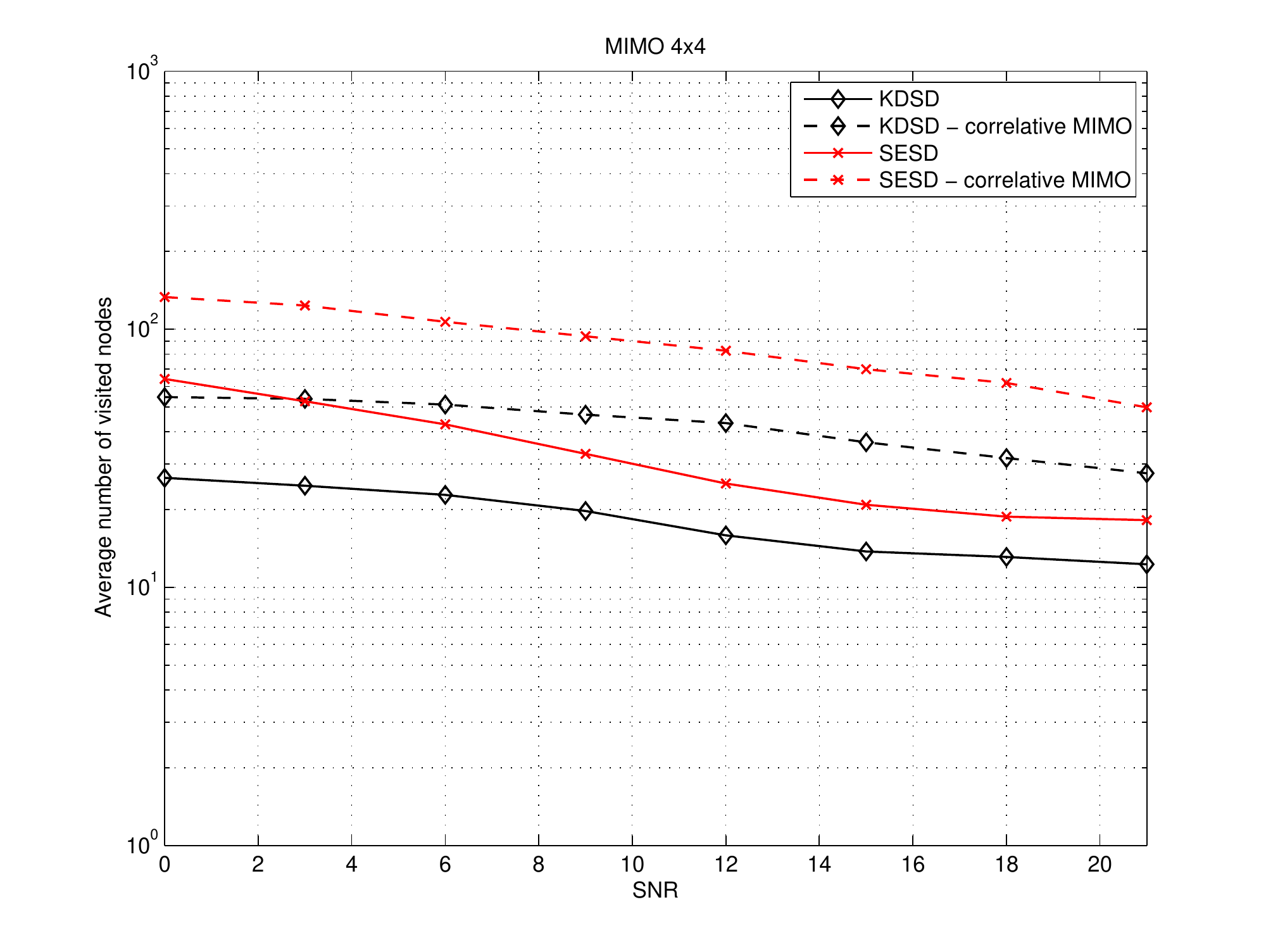}
\par\end{centering}

\caption{\label{fig:avg-num-nodes-1}Average number of visited nodes as function
of average signal-to-noise ratio. $4$-QAM system with $K=4$, $N=4$.}
\end{figure}

\noindent 
\begin{figure}
\noindent \begin{centering}
\includegraphics[width=1\columnwidth]{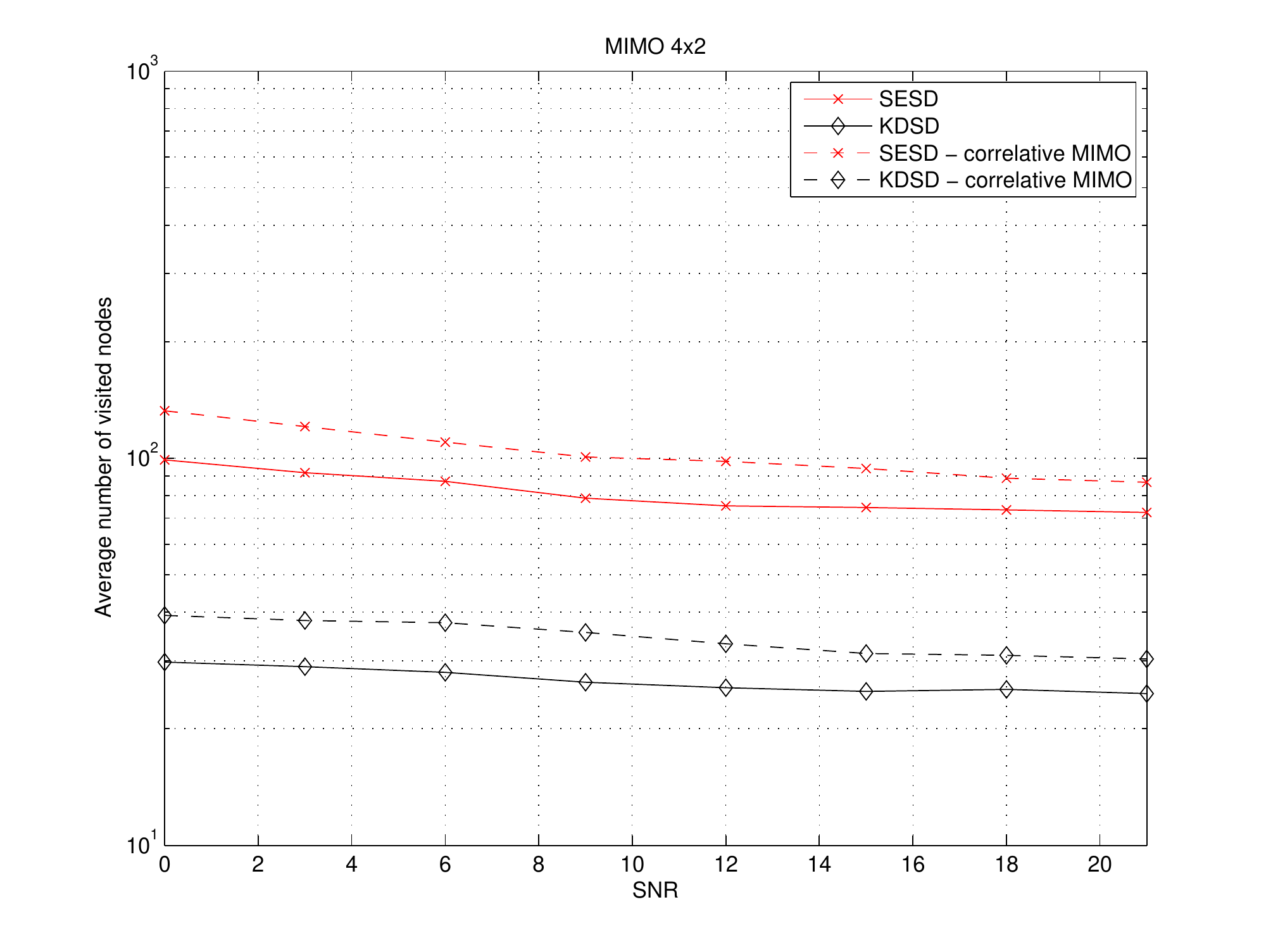}
\par\end{centering}

\caption{\label{fig:avg-num-nodes-2}Average number of visited nodes as function
of average signal-to-noise ratio. $16$-QAM system with $K=2$, $N=4$.}
\end{figure}

\section{Conclusions\label{sec:conclusions}}

We have proposed an enhanced version of the SD, namely KSD, that presents
a lower computational complexity measured in terms of average number
of visited nodes, w.r.t classic SD implementation. The reduction in
complexity is possible because an additional cost function is considered
in the standard tree-search based SD. The cost function is based on
the dominance conditions that allows to take a decision when multiantenna
interference is not too strong. Therefore the KSD has all the features
of any SD algorithm and has always better performances. Numerical
simulations show that for MIMO systems, both with independent and
correlated fading statistics, the dominance conditions effectively
reduce the computational complexity of the SD.

\section{\label{proof:k-ddiff}Proof of proposition \ref{pro:k-ddiff}}

We explicitly write the discrete difference as:
\begin{multline}
\Delta_{k}f\left(\mathbf{x};\hat{\mathbf{x}}\right)=\\
-2\Re\left\{ \left(x_{k}-\hat{x}_{k}\right)^{*}\mathbf{h}_{k}^{H}\mathbf{y}\right\} +\mathbf{x}^{H}\mathbf{H}^{H}\mathbf{H}\mathbf{x}-\hat{\mathbf{x}}^{H}\mathbf{H}^{H}\mathbf{H}\hat{\mathbf{x}}\label{eq:proof-ddiff-1}
\end{multline}
The term $\mathbf{x}^{H}\mathbf{H}^{H}\mathbf{H}\mathbf{x}-\hat{\mathbf{x}}^{H}\mathbf{H}^{H}\mathbf{H}\hat{\mathbf{x}}$
is a real scalar, so we can apply the conjugate-transpose operator
with no change to obtain
\begin{multline}
\mathbf{x}^{H}\mathbf{H}^{H}\mathbf{H}\mathbf{x}-\hat{\mathbf{x}}^{H}\mathbf{H}^{H}\mathbf{H}\hat{\mathbf{x}}=\\
\Re\left\{ \sum_{i}\sum_{j}x_{i}^{*}\mathbf{h}_{i}^{H}\mathbf{h}_{j}x_{j}-\sum_{m}\sum_{n}\hat{x}_{m}^{*}\mathbf{h}_{m}^{H}\mathbf{h}_{n}\hat{x}_{n}\right\} =\\
2\Re\left\{ \left(x_{k}-\hat{x}_{k}\right)^{*}\left[\mathbf{h}_{k}^{H}\mathbf{y}-\sum_{i\neq k}x_{i}\mathbf{h}_{k}^{H}\mathbf{h}_{i}\right]\right\} \\
+\left(\left|x_{k}\right|^{2}-\left|\hat{x}_{k}\right|^{2}\right)\mathbf{h}_{k}^{H}\mathbf{h}_{k}\label{eq:proof-ddiff-2}
\end{multline}
By substituting \eqref{eq:proof-ddiff-2} into \eqref{eq:proof-ddiff-1}
we can write the discrete difference as stated by the proposition.

\section{\label{proof:suff-cond}Proof of Proposition \ref{pro:suff-cond}}

The sign of the $k$th discrete difference for $\mathbf{x}$ (w.r.t.
its unique adjacent vector $\hat{\mathbf{x}}$ along the $k$th coordinate)
is determined by \eqref{eq:equil-condition}, reported in the following
for convenience:
\begin{equation}
x_{k}=\mathrm{sign}\left[\mathbf{h}_{k}^{T}\mathbf{y}-\sum_{i\neq k}x_{i}\mathbf{h}_{k}^{T}\mathbf{h}_{i}\right].\label{eq:equil-condition-Appendix}
\end{equation}
When the sufficient condition \eqref{eq:suff-cond} holds, the first
term in r.h.s. of \eqref{eq:equil-condition-Appendix} is dominant
over the sum representing the second term, independently on $x_{i}$,
$i\neq k$. In such a case \eqref{eq:equil-condition-Appendix} reduces
to 
\[
x_{k}=\mathrm{sign}\left[\mathbf{h}_{k}^{T}\mathbf{y}\right],
\]
that is the $k$th discrete difference depends only on $x_{k}$, as
stated by the proposition.

\bibliographystyle{IEEEtran}
\bibliography{IEEEabrv,myabbrv,king-decoder-journal,king-decoderMPSK}

\begin{thebibliography}{10}
\providecommand{\url}[1]{#1}
\csname url@samestyle\endcsname
\providecommand{\newblock}{\relax}
\providecommand{\bibinfo}[2]{#2}
\providecommand{\BIBentrySTDinterwordspacing}{\spaceskip=0pt\relax}
\providecommand{\BIBentryALTinterwordstretchfactor}{4}
\providecommand{\BIBentryALTinterwordspacing}{\spaceskip=\fontdimen2\font plus
\BIBentryALTinterwordstretchfactor\fontdimen3\font minus
  \fontdimen4\font\relax}
\providecommand{\BIBforeignlanguage}[2]{{%
\expandafter\ifx\csname l@#1\endcsname\relax
\typeout{** WARNING: IEEEtran.bst: No hyphenation pattern has been}%
\typeout{** loaded for the language `#1'. Using the pattern for}%
\typeout{** the default language instead.}%
\else
\language=\csname l@#1\endcsname
\fi
#2}}
\providecommand{\BIBdecl}{\relax}
\BIBdecl

\bibitem{Biglieri2007}
E.~Biglieri, \emph{{MIMO wireless communications}}.\hskip 1em plus 0.5em minus
  0.4em\relax Cambridge University Press, 2007.

\bibitem{Verdu1998}
S.~Verdú, \emph{Multiuser Detection}.\hskip 1em plus 0.5em minus 0.4em\relax
  Cambrige University Press, 1998.

\bibitem{Proakis2000}
J.~Proakis, \emph{Digital Communications}, 4th~ed.\hskip 1em plus 0.5em minus
  0.4em\relax McGraw-Hill, 2000.

\bibitem{Luo2004}
J.~Luo, K.~Pattipati, P.~Willett, and G.~Levchuk, ``Fast optimal and suboptimal
  any-time algorithms for {CDMA} multiuser detection based on branch and
  bound,'' \emph{{IEEE} Trans. Commun.}, vol.~52, no.~4, pp. 632--642, Apr.
  2004.

\bibitem{Mow2003}
W.~H. Mow, ``Universal lattice decoding: principle and recent advances,''
  \emph{Wireless Communications and Mobile Computing}, vol.~3, no.~5, pp.
  553--569, 2003.

\bibitem{Ekroot1996}
L.~Ekroot and S.~Dolinar, ``A* decoding of block codes,'' \emph{{IEEE} Trans.
  Commun.}, vol.~44, no.~9, pp. 1052--1056, Sep. 1996.

\bibitem{Murugan2006}
A.~D. Murugan, H.~E. Gamal, M.~O. Damen, and G.~Caire, ``A unified framework
  for tree search decoding: rediscovering the sequential decoder,''
  \emph{{IEEE} Trans. Inf. Theory}, vol.~52, no.~3, pp. 933--953, Mar. 2006.

\bibitem{Fincke1985}
U.~Fincke and M.~Pohst, ``Improved methods for calculating vectors of short
  length in a lattice, including a complexity analysis,'' \emph{Mathematics of
  Computation}, vol.~44, no. 170, pp. 463--471, 1985.

\bibitem{Schnorr1994}
C.~P. Schnorr and M.~Euchner, ``{Lattice basis reduction: Improved practical
  algorithms and solving subset sum problems},'' \emph{Mathematical
  Programming}, vol.~66, pp. 181--199, 1994.

\bibitem{Mow1994}
W.~H. Mow, ``Maximum likelihood sequence estimation from the lattice
  viewpoint,'' \emph{{IEEE} Trans. Inf. Theory}, vol.~40, no.~5, pp. 1591
  --1600, Sep. 1994.

\bibitem{Viterbo1999}
E.~Viterbo and J.~Boutros, ``A universal lattice code decoder for fading
  channels,'' \emph{{IEEE} Trans. Inf. Theory}, vol.~45, no.~5, pp. 1639--1642,
  1999.

\bibitem{Damen2000}
O.~Damen, A.~Chkeif, and J.~C. Belfiore, ``Lattice code decoder for space-time
  codes,'' \emph{{IEEE} Commun. Lett.}, vol.~4, no.~5, pp. 161--163, May 2000.

\bibitem{Agrell2002}
E.~Agrell, T.~Eriksson, A.~Vardy, and K.~Zeger, ``Closest point search in
  lattices,'' \emph{{IEEE} Trans. Inf. Theory}, vol.~48, no.~8, pp. 2201--2214,
  Aug. 2002.

\bibitem{Damen2003}
M.~O. Damen, H.~El~Gamal, and G.~Caire, ``On maximum-likelihood detection and
  the search for the closest lattice point,'' \emph{{IEEE} Trans. Inf. Theory},
  vol.~49, no.~10, pp. 2389--2402, Oct. 2003.

\bibitem{Cui2005}
T.~Cui and C.~Tellambura, ``An efficient generalized sphere decoder for
  rank-deficient {MIMO} systems,'' \emph{{IEEE} Commun. Lett.}, vol.~9, no.~5,
  pp. 423--425, May 2005.

\bibitem{Chang2007}
X.-W. Chang and X.~Yang, ``An efficient tree search decoder with column
  reordering for underdetermined {MIMO} systems,'' in \emph{{P}roc. {IEEE}
  {G}lobal {T}elecommunications {C}onference (GLOBECOM '07)}, Nov. 26--30,
  2007, pp. 4375--4379.

\bibitem{Wong2007}
K.-K. Wong and A.~Paulraj, ``Efficient near maximum-likelihood detection for
  underdetermined mimo antenna systems using a geometrical approach,''
  \emph{EURASIP J. Wirel. Commun. Netw.}, vol. 2007, pp. 10:1--10:13, Oct.
  2007.

\bibitem{Hassibi2005}
B.~Hassibi and H.~Vikalo, ``{On the sphere-decoding algorithm I. Expected
  complexity},'' \emph{{IEEE} Trans. Signal Process.}, vol.~53, no.~8, pp.
  2806--2818, Aug. 2005.

\bibitem{Vikalo2005}
H.~Vikalo and B.~Hassibi, ``{On the sphere-decoding algorithm II.
  Generalizations, second-order statistics, and applications to
  communications},'' \emph{{IEEE} Trans. Signal Process.}, vol.~53, no.~8, pp.
  2819--2834, Aug. 2005.

\bibitem{Jalden2005}
J.~Jald\'{e}n and B.~Ottersten, ``On the complexity of sphere decoding in
  digital communications,'' \emph{{IEEE} Trans. Signal Process.}, vol.~53,
  no.~4, pp. 1474--1484, Apr. 2005.

\bibitem{Guo2006}
Z.~Guo and P.~Nilsson, ``Algorithm and implementation of the {K-best} sphere
  decoding for {MIMO} detection,'' \emph{{IEEE} J. Sel. Areas Commun.},
  vol.~24, no.~3, pp. 491--503, Mar. 2006.

\bibitem{Odling2000}
P.~Odling, H.~Eriksson, and P.~Borjesson, ``Making mlsd decisions by
  thresholding the matched filter output,'' \emph{{IEEE} Trans. Commun.},
  vol.~48, no.~2, pp. 324 --332, Feb. 2000.

\bibitem{Axehill2008}
D.~Axehill, F.~Gunnarsson, and A.~Hansson, ``{A Low-Complexity High-Performance
  Preprocessing Algorithm for Multiuser Detection Using Gold Sequences},''
  \emph{{IEEE} Trans. Signal Process.}, vol.~56, no.~9, pp. 4377--4385, Sep.
  2008.

\bibitem{Romano2009}
G.~Romano, F.~Palmieri, P.~{Salvo Rossi}, and D.~Mattera, ``{A tree-search
  algorithm for ML decoding in underdetermined MIMO systems},'' in \emph{{Proc.
  International Symposium on Wireless Communications Systems} (ISWCS '09)},
  Siena, Italy, Sept. 2009, pp. 662 -- 666.

\bibitem{Romano2010}
G.~Romano, D.~Ciuonzo, P.~{Salvo Rossi}, and F.~Palmieri, ``Tree-search {ML}
  detection for underdetermined {MIMO} systems with {$M$-PSK} constellations,''
  in \emph{{Proc. International Symposium on Wireless Communications Systems}
  (ISWCS '10)}, Sept. 2010, pp. 102--106.

\bibitem{Hanzo2006}
L.~Hanzo and D.~T. Keller, \emph{OFDM and MC-CDMA}.\hskip 1em plus 0.5em minus
  0.4em\relax John Wiley \& Sons, 2006.

\bibitem{Studer2008}
C.~Studer, A.~Burg, and H.~B\"{o}lcskei, ``{Soft-output sphere decoding:
  algorithms and VLSI implementation},'' \emph{{IEEE} J. Sel. Areas Commun.},
  vol.~26, no.~2, pp. 290--300, 2008.

\bibitem{Hopfield1982}
J.~J. Hopfield, ``{Neural networks and physical systems with emergent
  collective computational abilities},'' in \emph{Proceedings of the National
  Academy of Sciences of the United States of America}, vol.~79, no.~8, 1982,
  pp. 2554--2558.

\bibitem{Hopfield1984}
------, ``{Neurons with graded response have collective computational
  properties like those of two-state neurons},'' in \emph{Proceedings of the
  National Academy of Sciences of the United States of America}, vol.~81,
  no.~10, 1984, pp. 3088--3092.

\bibitem{Miyajima1993}
T.~Miyajima, T.~Hasegawa, and M.~Haneish, ``On the multiuser detection using a
  neural network in code division multiple access communications,'' \emph{IEICE
  Transactions on Communications}, vol. E76-B, pp. 961--968, 1993.

\bibitem{Kechriotis1993}
G.~I. Kechriotis and E.~S. Manolakos, ``Implementing the optimal {CDMA}
  multiuser detector with hopfield neural networks,'' in \emph{Proc. Int. Wkshp
  Applicat. Neural Networks Telecommun.}, Princeton, NJ, Oct 1993, pp. 60--67.

\bibitem{Kechriotis1996}
------, ``Hopfield neural network implementation of the optimal {CDMA}
  multiuser detector,'' \emph{{IEEE} Trans. Neural Netw.}, vol.~7, no.~1, pp.
  131--141, Jan. 1996.

\bibitem{Kechriotis1996a}
------, ``A hybrid digital signal processing-neural network {CDMA} multiuser
  detection scheme,'' \emph{{IEEE} Trans. Circuits Syst. {II}}, vol.~43, no.~2,
  pp. 96--104, Feb. 1996.

\bibitem{Sgraja2001}
C.~Sgraja, W.~Teich, A.~Engelhart, and J.~Lindner, ``Multiuser/multisubchannel
  detection based on recurrent neural network structures for linear modulation
  schemes with general complex-valued symbol alphabet,'' in \emph{Proc. COST
  262 Workshop on Multiuser Detection in Spread Spectrum Communications}, Jan
  2001.

\bibitem{Engelhart2002}
A.~Engelhart, W.~G. Teich, J.~Lindner, G.~Jeney, S.~Imre, and L.~Pap, ``A
  survey of multiuser/multisubchannel detection schemes based on recurrent
  neural networks,'' \emph{Wireless Communications and Mobile Computing},
  vol.~2, no.~3, pp. 269--284, 2002.

\bibitem{Marcus1989}
C.~M. Marcus and R.~M. Westervelt, ``Dynamics of iterated-map neural
  networks,'' \emph{Phys. Rev. A}, vol.~40, no.~1, pp. 501--504, Jul 1989.

\bibitem{Sun2009}
Y.~Sun, ``A family of likelihood ascent search multiuser detectors: an upper
  bound of bit error rate and a lower bound of asymptotic multiuser
  efficiency,'' \emph{{IEEE} Trans. Commun.}, vol.~57, no.~6, pp. 1743--1752,
  Jun. 2009.

\bibitem{Chizhik2002}
D.~Chizhik, G.~Foschini, M.~Gans, and R.~Valenzuela, ``{Keyholes, correlations,
  and capacities of multielement transmit and receive antennas},'' \emph{{IEEE}
  Trans. Wireless Commun.}, vol.~1, no.~2, pp. 361--368, Apr. 2002.

\bibitem{Gesbert2002}
D.~Gesbert, H.~B\"{o}lcskei, D.~Gore, and A.~Paulraj, ``{Outdoor MIMO wireless
  channels: models and performance prediction},'' \emph{{IEEE} Trans. Commun.},
  vol.~50, no.~12, pp. 1926--1934, Dec. 2002.

\bibitem{Zelst2002}
A.~V. Zelst and J.~S. Hammerschmidt, ``A single coefficient spatial correlation
  model for multiple-input multiple-output ({MIMO}) radio channels,'' in
  \emph{{Proc. of the XXVIIth General Assembly of the International Union of
  Radio Science {URSI}}}, 2002, pp. 1--4.

\end{thebibliography}

\end{document}